\RequirePackage{amsmath}
\documentclass[runningheads]{llncs}
\usepackage{multicol}
\usepackage{amsfonts}
\usepackage{amssymb}
\usepackage{graphicx}
\usepackage{epic}
\usepackage{eepic}
\usepackage{epsfig,float}
\usepackage{verbatim}
\usepackage{pdfsync}
\usepackage{gastex}
\usepackage[lowtilde]{url}
\usepackage{enumitem}
\usepackage{url} 
\usepackage{hyperref} 

\usepackage{bbm}
\usepackage{xcolor}

\usepackage{thmtools}
\usepackage{thm-restate}

\usepackage{xcolor}

\usepackage[utf8]{inputenc}\usepackage[T1]{fontenc}

\renewcommand{\le}{\leqslant}
\renewcommand{\ge}{\geqslant}

\newcommand{\eps}{\varepsilon}
\newcommand{\emp}{\emptyset}

\newcommand{\Sig}{\Sigma}
\newcommand{\sig}{\sigma}
\newcommand{\noin}{\noindent}

\newcommand{\bi}{\begin{itemize}}
\newcommand{\ei}{\end{itemize}}
\newcommand{\be}{\begin{enumerate}}
\newcommand{\ee}{\end{enumerate}}
\newcommand{\bd}{\begin{description}}
\newcommand{\ed}{\end{description}}
\newcommand{\bq}{\begin{quote}}
\newcommand{\eq}{\end{quote}}

\newcommand{\cD}{{\mathcal D}}

\newcommand{\cN}{{\mathcal N}}

\newcommand{\one}{{\mathbf 1}}

\newcommand{\e}{\emph}

\renewcommand{\phi}{\varphi}

\newcommand{\tr}[1]{\stackrel{#1}{\longrightarrow}}

\title{Most Complex Deterministic Union-Free Regular Languages\thanks{This work was supported by the Natural Sciences and Engineering Research Council of Canada 
grant No.~OGP0000871.}
}

\author{Janusz A. Brzozowski\inst{1} 
\and Sylvie Davies\inst{2}
}

\authorrunning{J. A. Brzozowski, 
S. Davies
 }  

\titlerunning{Most Complex Union-Free Languages}

\institute{David R. Cheriton School of Computer Science, University of Waterloo \\
Waterloo, ON, Canada N2L 3G1\\
{\tt brzozo@uwaterloo.ca}
\and
Department of Pure Mathematics, University of Waterloo \\
Waterloo, ON, Canada N2L 3G1\\
{\tt sldavies@uwaterloo.ca}
}

\begin{document}
\maketitle

\begin{abstract}
A regular language $L$ is \emph{union-free} if it can be represented by a regular expression without the union operation. A union-free language is \emph{deterministic} if it can be accepted 
by a deterministic \emph{one-cycle-free-path} finite automaton; this is an automaton which has one final state and  exactly one cycle-free path from any state to the final state. 
Jir\'askov\'a and Masopust proved that the state complexities of the basic operations reversal, star, product, and boolean operations in deterministic union-free languages are exactly the same as those in the class of all regular languages. 
To prove that the bounds are met they used five  types of automata, involving eight  types of transformations of the set of states of the automata. 
We show that for each $n\ge 3$ there exists one  ternary witness  of state complexity $n$ that meets the bound for reversal and product.
Moreover, the restrictions of this witness to  binary alphabets meet the bounds for star and boolean operations. 
We also show that the tight upper bounds on the state complexity of binary operations that take arguments over different alphabets are the same as those for arbitrary regular languages. 
Furthermore, we prove that the maximal syntactic semigroup of a union-free language has $n^n$ elements, as in the case of regular languages, and 
that the maximal state complexities of atoms of union-free languages  are the same as those for regular languages. 
Finally, we prove that there exists a most complex union-free language that meets the bounds for all these complexity measures. 
Altogether this proves that the complexity measures above cannot distinguish  union-free languages from regular languages. 
\medskip

\noin
{\bf Keywords:}
atom, boolean operation, concatenation, different alphabets, most complex, one-cycle-free-path, regular, reversal,  star, state complexity, syntactic semigroup, transition semigroup, union-free
\end{abstract}

\section{Introduction} 
Formal definitions are postponed until Section~\ref{sec:preliminaries}.

The class of regular languages over a finite alphabet $\Sigma$ is the smallest class of languages containing the empty language $\emp$, the language $\{\eps\}$, where $\eps$ is the empty word, and the \emph{letter languages} $\{a\}$ for each $a\in \Sigma$, and closed under the operations of union, concatenation, and (Kleene) star. 
Hence each regular language can be written as a finite expression involving the above basic languages and operations.
An expression defining a regular language in this way is called a \emph{regular expression}.
Because regular languages are also closed under complementation, we may also consider regular expressions that allow complementation, which are called \e{extended regular expressions}.
In this paper we deal exclusively with regular languages.

A natural question is: what kind of languages are defined if one of the operations in the definitions given above is missing? If the star operation is removed from the extended regular expressions we get the well known \emph{star-free} languages~\cite{BrSz15b,McPa71,Sch65}, which have been extensively studied. Less attention was given to classes defined by removing an operation from ordinary regular expressions, but recently language classes defined without union or concatenation have been studied.

If we remove some operations from regular expressions, we obtain the following classes of languages:
\bd
\item[Union only] subsets of $\{\eps\} \cup \Sig$. 
\item[Concatenation only] $\emp$ and $\{w\}$ for each $w \in \Sig^*$.
\item[Star only] $\emp$, $\{\eps\}$, $\{a\}$ for each $a\in \Sigma$, and $\{a\}^*$ for each $a\in \Sigma$.
\item[Union and Concatenation] Finite languages.
\item[Concatenation and Star] These are the \emph{union-free} languages that constitute the main topic of this paper. 
\item[Union and Star] These are the \emph{concatenation-free} languages that were studied in~\cite{CDE01,HoKu17,KuWe17}.
\ed

Union-free regular languages were first considered by Brzozowski~\cite{Brz62} in 1962 under the name \emph{star-dot} regular languages, where \emph{dot} stands for concatenation.
He proved that every regular language is  a union of union-free languages~\cite[p.~216, Theorem 9.5]{Brz62}\footnote{Terminology changed to that of the present paper.}.
Much more recently, in 2001, Crvenkovi\'c, Dolinka and \'Esik~\cite{CDE01} studied equations satisfied by union-free regular languages, and proved that the class of 
these languages cannot be axiomatized by a finite set of equations. This is also known to be true for the class of all regular languages.
In 2006 Nagy studied union-free languages in detail and characterized them in terms of nondeterministic finite automata (NFAs) recognizing them~\cite{Nag06}, which  he called \emph{one-cycle-free-path} NFAs.
In 2009 minimal union-free decompositions of regular languages were studied in~\cite{AfGo09} by Afonin and Golomazov. They also presented a new algorithm for deciding whether a given deterministic finite automaton (DFA) accepts a union-free language.
Decompositions of regular languages in terms of union-free languages were further  studied  by Nagy in 2010~\cite{Nag10}.
The state complexities of operations on union-free languages were examined in 2011 by 
Jir\'askov\'a and Masopust~\cite{JiMa11}, who proved that the state complexities of basic operations on these languages are the same as those in the class of all regular languages.
It was shown in~\cite{JiMa11} that the class of languages defined by DFAs with the one-cycle-free-path property is a proper subclass of that defined by  one-cycle-free-path NFAs; the former class is called the class of \emph{deterministic union-free} languages.
In~2012 Jir\'askov\'a and Nagy~\cite{JiNa12} proved that the class of finite unions of  deterministic union-free languages is a proper subclass of the class of regular languages.
They also showed that every deterministic union-free language is accepted by a special kind of a one-cycle-free-path DFA called a \emph{balloon} DFA.
A summary of the properties of union-free languages was presented in 2017 in~\cite{HoKu17}.

\section{Preliminaries}
\label{sec:preliminaries}
Let $L$ be a regular language.
We define the \e{alphabet} of $L$ to be the set of letters which appear \e{at least once} in a word of $L$.
For example, consider the language $L = \{a,ab,ac\}$ and the subset $K = \{a,ac\}$; we say $L$ has alphabet $\{a,b,c\}$ and $K$ has alphabet $\{a,c\}$.

A \emph{deterministic finite automaton (DFA)} is a quintuple
$\cD=(Q, \Sigma, \delta, q_0,F)$, where
$Q$ is a finite non-empty set of \emph{states},
$\Sig$ is a finite non-empty \emph{alphabet},
$\delta\colon Q\times \Sig\to Q$ is the \emph{transition function},
$q_0\in Q$ is the \emph{initial} state, and
$F\subseteq Q$ is the set of \emph{final} states.
We extend $\delta$ to functions $\delta\colon Q\times \Sig^*\to Q$ and $\delta\colon 2^Q\times \Sig^*\to 2^Q$ as usual (where $2^Q$ denotes the set of all subsets of $Q$).
A~DFA $\cD$ \emph{accepts} a word $w \in \Sigma^*$ if ${\delta}(q_0,w)\in F$. The \e{language accepted by $\cD$} is the set of all words accepted by $\cD$, and is denoted by $L(\cD)$. If $q$ is a state of $\cD$, then the language $L_q(\cD)$ of $q$ is the language accepted by the DFA $(Q,\Sigma,\delta,q,F)$. 
A state is \emph{empty} (or \emph{dead} or a \emph{sink state}) if its language is empty. Two states $p$ and $q$ of $\cD$ are \emph{equivalent} if $L_p(\cD) = L_q(\cD)$. 
A state $q$ is \emph{reachable} if there exists $w\in\Sig^*$ such that $\delta(q_0,w)=q$.
A DFA $\cD$ is \emph{minimal} if it has the smallest number of states among all DFAs accepting $L(\cD)$.
We say a DFA has a \e{minimal alphabet} if its alphabet is equal to the alphabet of $L(\cD)$.
It is well known that a DFA with a minimal alphabet is minimal if and only if all of its states are reachable and no two states are equivalent.

A \emph{nondeterministic finite automaton (NFA)} is a quintuple
$\mathcal{N}=(Q, \Sigma, \delta, I,F)$, where
$Q$,
$\Sigma$ and $F$ are as in a DFA, 
$\delta\colon Q\times \Sigma\to 2^Q$, and
$I\subseteq Q$ is the \emph{set of initial states}. 
Each triple $(p,a,q)$ with $p,q\in Q$, $a\in\Sig$ is a \emph{transition}  if $q\in \delta(p,a)$.
A sequence $((p_0,a_0,q_0), (p_1,a_1,q_1), \dots, (p_{k-1},a_{k-1},q_{k-1}))$
 of transitions, where $p_{i+1}=q_i$ for $i=0, \dots, k-2$ is a \emph{path} in $\cN$.
The word $a_0a_1\cdots a_{k-1}$ is the word \emph{spelled} by the path. 
A word $w$ is \emph{accepted} by $\cN$ is there exists a path with $p_0\in I$ and $q_{k-1}\in F$ that spells $w$.
If $q\in \delta(p,a)$ we also use the notation $p \xrightarrow{a} q$. We extend this notation also to words, and write 
$p \xrightarrow{w} q$ for $w\in\Sig^*$.

The \emph{state complexity} of a regular language $L$, denoted by $\kappa(L)$, is the number of states in the minimal DFA accepting $L$. 
Henceforth we frequently refer to state complexity as simply \e{complexity}, and we denote a language of complexity $n$ by $L_n$, and a DFA with $n$ states by $\cD_n$.

The \emph{state complexity} of a regularity-preserving \e{unary} operation $\circ$ on regular languages
is the maximal value of $\kappa(L^\circ)$, expressed as a function of one parameter $n$, where $L$ varies over all regular languages with complexity at most $n$.
For example, the state complexity of the reversal operation is $2^n$; it is known that if $L$ has complexity at most $n$, then $\kappa(L^R) \le 2^n$, and furthermore this upper bound is tight in the sense that for each $n \ge 1$ there exists a language $L_n$ such that $\kappa(L_n^R) = 2^n$.
In general, to show that an upper bound on $\kappa(L^\circ)$ is tight, we need to exhibit a sequence $(L_n \mid n\ge k)=(L_k,L_{k+1},\dots)$, called a \emph{stream}, of languages of each complexity $n \ge k$ (for some small constant $k$) that meet this upper bound. 
Often we are not interested in the special-case behaviour of the operation that may occur at very small values of $n$; the parameter $k$ allows us to ignore these small values and simplify the statements of results.

The \emph{state complexity} of a regularity-preserving \e{binary} operation $\circ$ on regular languages is the maximal value of $\kappa(L' \circ L)$, epxressed as a function of two parameters $m$ and $n$, where $L'$ varies over all regular languages of complexity at most $m$ and $L$ varies over all regular languages of complexity at most $n$.
In this case, to show an upper bound on the state complexity is tight, we need to exhibit two  classes $(L'_{m,n}\mid  m \ge h, n\ge k)$ and $( L_{m,n} \mid m \ge h, n \ge k)$ of languages meeting the bound;
the notation $L'_{m,n}$ and $L_{m,n}$ implies that  $L'_{m,n}$ and $L_{m,n}$ depend on both $m$ and $n$. 
However, in most cases studied in the literature, it is enough to use witness streams $(L'_m \mid m \ge  h)$ and $(L_n \mid n\ge k)$, where  $L'_m$ is independent of $n$ and $L_n$ is independent of $m$. 

For binary operations we consider two types of state complexity: \e{restricted} and \e{unrestricted} state complexity. For restricted state complexity the operands of the binary operations are required to have the same alphabet. For unrestricted state complexity the alphabets of the operands may differ. See~\cite{BrSi17b} for more details.

Sometimes the same stream can be used for both operands of a binary operation, but this is not always possible. For example,  for boolean operations when $m=n$, the state complexity of 
$L_n\cup L_n=L_n$ is $n$, whereas the upper bound is $mn=n^2$.
However, in many cases the second language  is a "dialect" of the first, that is, it ``differs only slightly'' from the first. 
A \emph{dialect}  of $L_n(\Sig)$ is a language obtained from $L_n(\Sig)$  by  deleting some letters of $\Sigma$ in the words of $L_n(\Sig)$ -- by this we mean that words containing these letters are deleted -- or replacing them by letters of another alphabet $\Sigma'$.
In this paper we consider only the cases where $\Sig=\Sig'$,
and we encounter only two types of dialects:
\be
\item
A dialect in which some letters were deleted; for example, $L_n(a,b)$ is a dialect of $L_n(a,b,c)$ with $c$ deleted, and $L_n(a,-,c)$ is a dialect with $b$ deleted.
\item
 A dialect in which the roles of two letters are exchanged; for example, $L_n(b,a)$ is such a dialect of $L_n(a,b)$.
\ee
These two types of dialects can be combined, for example, in $L_n(a,-,b)$ the letter $c$ is deleted, and $b$ plays the role that $c$ played originally.
The notion of dialects also extends to DFAs; for example, if $\cD_n(a,b,c)$ recognizes $L_n(a,b,c)$ then $\cD_n(a,-,b)$ recognizes the dialect $L_n(a,-,b)$.

We use $Q_n=\{0,\dots,n-1\}$ as our basic set with $n$ elements.
A \emph{transformation} of $Q_n$ is a mapping $t\colon Q_n\to Q_n$.
The \emph{image} of $q\in Q_n$ under $t$ is denoted by $qt$, and this notation is extended to subsets of $Q_n$.
The \e{preimage} of $q \in Q_n$ under $t$ the set $qt^{-1} = \{ p \in Q_n : pt = q\}$, and this notation is extended to subsets of $Q_n$ as follows: $St^{-1} = \{ p \in Q_n : pt \in S\}$.
The \emph{rank} of a transformation $t$ is the cardinality of $Q_nt$.
If $s$ and $t$ are transformations of $Q_n$, their composition is  denoted $st$ and we have $q(st) = (qs)t$ for $q \in Q_n$.
The $k$-fold composition $tt \dotsb t$ (with $k$ occurences of $t$) is denoted $t^k$, and for $S \subseteq Q_n$ we define $St^{-k} = S(t^k)^{-1}$.
Let $\mathcal{T}_{Q_n}$ be the set of all $n^n$ transformations of $Q_n$; then $\mathcal{T}_{Q_n}$ is a monoid under composition. 

For $k\ge 2$, a transformation 
$t$ of a set $P=\{q_0,q_1,\ldots,q_{k-1}\} \subseteq Q_n$ is a \emph{$k$-cycle}
if $q_0t=q_1, q_1t=q_2,\ldots,q_{k-2}t=q_{k-1},q_{k-1}t=q_0$.
This $k$-cycle is denoted by $(q_0,q_1,\ldots,q_{k-1})$, and leaves the states in $Q_n\setminus P$ unchanged.
A~2-cycle $(q_0,q_1)$ is called a \emph{transposition}.
A transformation  that sends state $p$ to $q$ and acts as the identity on the remaining states is denoted by $(p \to q)$.
The identity transformation is denoted by {\bf 1}.

Let $\mathcal{D} = (Q_{n}, \Sigma, \delta,  0, F)$ be a DFA. For each word $w \in \Sigma^*$, the transition function induces a transformation $\delta_w$ of $Q_n$ by  $w$: for all $q \in Q_n$, 
$q\delta_w = \delta(q, w).$ 
The set $T_{\mathcal{D}}$ of all such transformations by non-empty words is the \emph{transition semigroup} of $\mathcal{D}$ under composition.
Often we use the word $w$ to denote the transformation $t$ it induces; thus we write $qw$ instead of $q\delta_w$. 
We also write 
$w\colon t$ to mean that $w$ induces the transformation $t$.

The size of the \emph{syntactic semigroup} of a regular language is another measure of the complexity of the language~\cite{Brz13}.
Write $\Sig^+$ for $\Sig^* \setminus \{\eps\}$.
The 
\emph{syntactic congruence} of a language $L\subseteq \Sigma^*$ is defined on $\Sigma^+$ as follows:
For $x, y \in \Sigma^+,  x \,{\mathbin{\approx_L}}\, y $  if and only if  $wxz\in L  \Leftrightarrow wyz\in L$ for all  $w,z \in\Sigma^*.
$
The quotient set $\Sigma^+/ {\mathbin{\approx_L}}$ of equivalence classes of  ${\mathbin{\approx_L}}$ is a semigroup, the \emph{syntactic semigroup} $T_L$ of $L$.
The syntactic semigroup is isomorphic to  the \emph{transition semigroup} of the minimal DFA of $L$~\cite{Pin97}.

The (left) \e{quotient} of $L\subseteq \Sig^*$ by a word $w\in \Sig^*$ is the language $w^{-1}L=\{x : wx \in L\}$. It is well known that the number of quotients of a regular language is finite and equal to the state complexity of the language.

The atoms of a regular language are defined by a left congruence, where two words $x$ and $y$ are congruent whenever 
 $ux\in L$ if and only if  $uy\in L$ for all $u\in \Sigma^*$. 
 Thus $x$ and $y$ are congruent whenever  $x\in u^{-1}L$ if and only if $y\in u^{-1}L$
 for all $u\in \Sigma^*$.
 An equivalence class of this relation is an \emph{atom} of $L$~\cite{BrTa14}. 
Atoms can be expressed as non-empty intersections of complemented and uncomplemented quotients of $L$.
The number of atoms and their state complexities were suggested as measures of complexity of regular languages~\cite{Brz13}
because all quotients of a language and all quotients of its atoms are unions of atoms~\cite{BrTa13,BrTa14,Iva16}.

\section{Main Results}
The automata described in~\cite{Nag06} that characterize union-free languages are called there \emph{one-cycle-free-path} automata. They are defined by the property that there is only one final state and a unique cycle-free path from each state to the final state.
We are now ready to define a most complex deterministic one-cycle-free-path DFA and its most complex deterministic union-free language.

The most complex stream below meets all of our complexity bounds. However, our witness uses three letters for restricted product whereas~\cite{JiMa11} uses binary witnesses. The same shortcoming of most complex streams occurs in the case of regular languages~\cite{Brz13}; that seems to be the price of getting a witness for all operations rather than minimizing the alphabet for each operation.

\begin{definition}
\label{def:union-free}
For $n\ge 3$, let 
$\mathcal{D}_n=\mathcal{D}_n(a,b,c, d)=(Q_n,\Sigma,\delta_n, 0, \{n-1\})$, where $\Sig=\{a,b,c, d\}$, and
$\delta_n$ is defined by the transformations 
$a\colon (1,\dots,n-1)$, $b\colon (0,1)$,  $c\colon (1\to 0)$, and $d\colon \one$;
see Figure~\ref{fig:union-free}.
Let $L_n=L_n(a,b,c, d)$ be the language accepted by~$\mathcal{D}_n(a,b,c, d)$.
\end{definition}

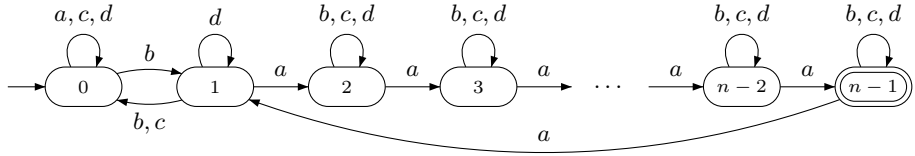
\begin{figure}[ht]
\unitlength 8.3pt
\begin{center}\begin{picture}(36,6)(0,4)
\gasset{Nh=1.8,Nw=3.5,Nmr=1.25,ELdist=0.4,loopdiam=1.5}
	{\scriptsize
\node(0)(1,7){0}\imark(0)
\node(1)(7,7){1}
\node(2)(13,7){2}
\node(3)(19,7){3}
}
\node[Nframe=n](3dots)(25,7){$\dots$}
	{\scriptsize
\node(n-2)(31,7){$n-2$}
	}
	{\scriptsize
\node(n-1)(37,7){$n-1$}\rmark(n-1)
	}
\drawedge[curvedepth= .8,ELdist=.4](0,1){$b$}	
\drawedge[curvedepth= .8,ELdist=.4](1,0){$b,c$}
\drawloop(0){$a,c,d$}
\drawedge[ELdist=.5](1,2){$a$}
\drawloop(1){$d$}
\drawloop(2){$b,c,d$}
\drawedge[ELdist=.5](n-2,n-1){$a$}
\drawloop(3){$b,c,d$}
\drawedge[ELdist=.5](3dots,n-2){$a$}
\drawedge[ELdist=.5](3,3dots){$a$}
\drawloop(n-2){$b,c,d$}
\drawedge[ELdist=.5](2,3){$a$}
\drawedge[curvedepth= 3.0,ELdist=-1.0](n-1,1){$a$}
\drawloop(n-1){$b,c,d$}
\end{picture}\end{center}
\caption{Most complex minimal one-cycle-free-path DFA  $\cD_n(a,b,c, d)$ of 
Definition~\ref{def:union-free}. }
\label{fig:union-free}
\end{figure}

The DFA of Definition 1 bears some similarities to the DFA for reversal in Fig. 6  in~\cite[p. 1650]{JiMa11}. 
It is evident that it is a one-cycle-free-path DFA.
Let $E= (a(b\cup c\cup d)^*)^{n-2}a$.
One verifies that 
\[
 \begin{array}{rl}
   	L_n  =  	 &[ (a \cup c \cup d) \cup b(d \cup E(b \cup c \cup d)^*a)^* (b \cup c)]^* \\
	  &b(d \cup E(b \cup c \cup d)^*a)^* E(b \cup c \cup 	d)^*.
 \end{array}
\]
Noting that $(E_1 \cup E_2\cup \cdots \cup E_k)^* =(E_1^*E_2^*\cdots E_k^*)^*$ for all regular expressions $E_i$, $i=1,\dots, k$, we obtain
a union-free expression for $L_n$.

\begin{theorem}[Most Complex Deterministic Union-Free Languages]
\label{thm:main}
For each $n\ge 3$, the DFA of Definition~\ref{def:union-free} is minimal and recognizes a deterministic union-free language. 
The stream $(L_n(a,b,c) \mid n \ge 3)$  with some dialect streams
is most complex in the class of deterministic union-free languages in the following sense:
\begin{enumerate}
\item
The syntactic semigroup of $L_n(a,b,c)$ has cardinality $n^n$, and at least 
three letters are required to reach this bound.
\item
Each quotient of $L_n(a,b)$ has complexity $n$.
\item
The reverse of $L_n(a,b,c)$ has complexity $2^n$.
Moreover, $L_n(a,b,c)$ has $2^n$ atoms. 
\item
Each atom  $A_S$  of $L_n(a,b,c)$ has maximal complexity:
\begin{equation*}
	\kappa(A_S) =
	\begin{cases}
		2^n-1, 			& \text{if $S\in \{\emptyset,Q_n\}$;}\\
		1+ \sum_{x=1}^{|S|}\sum_{y=1}^{n-|S|} \binom{n}{x}\binom{n-x}{y},
		 			& \text{if $\emptyset \subsetneq S \subsetneq Q_n$.}
		\end{cases}
\end{equation*}
\item
The star of $L_n(a,b)$ has complexity $2^{n-1}+2^{n-2}$.
\item 
	\begin{enumerate}
	\item
	Restricted  product:
	$\kappa(L_m(a,b,c) L_n(a,b,c)) = (m-1)2^n+2^{n-1}$.
	\item
	Unrestricted  product:
    $\kappa(L_m(a,b,c)L_n(a,b,c,d)) = m2^n+2^{n-1}$.
	\end{enumerate} 
\item 
	\begin{enumerate}
	\item
	Restricted boolean operations: For $(m,n) \neq (3,3)$,
	$\kappa(L_m(a,b)\circ L_n(b,a)) = mn$ for all binary boolean operations $\circ$ that depend on both arguments.
	\item
	Additionally, when $m\neq n$, 
	$\kappa(L_m(a,b)\circ L_n(a,b)) = mn$.
	\item
	 Unrestricted boolean operations ($\oplus$ denotes symmetric difference):
	 $$
	 \begin{cases}
	 \kappa(L_m(a,b,-,c) \circ L_n(b,a,-,d))=(m+1)(n+1) \text{ if } \circ\in \{\cup,\oplus\},\\
	\kappa(L_m(a,b,-,c) \setminus L_n(b,a))=mn+n,\\
 	L_m(a,b) \cap L_n(b,a)=mn.
	\end{cases}
	$$
	\end{enumerate}
\end{enumerate}
All of these bounds are maximal for deterministic union-free languages.
\end{theorem}
\begin{proof}
Only state 0 accepts $ba^{n-2}$, and the shortest word accepted by state $q$, $1\le q\le n-1$, is $a^{n-1-q}$. Hence all the states are distinguishable, and $\cD_n$ is minimal.
We  noted above that it recognizes a deterministic union-free language.

\be
\item
It is well known that the three transformations $a'\colon (0,\dots n-1)$, $b\colon (0,1)$, and $c\colon (1\to 0)$ generate all $n^n$ transformations of $Q_n$. We have $b$ and $c$ in $\cD_n$, and $a'$ is generated by $ab$. Hence our semigroup is maximal.

\item
This is easily verified.
\item
By~\cite{BrTa14} the  number of atoms is the same as the complexity of the reverse.
By~\cite{SWY04} the complexity of the reverse is $2^n$.
\item
The proof in~\cite{BrDa15} applies here as well.
\item 
We construct  an NFA for $(L_n(a,b))^*$ by taking $\cD_n(a,b)$ and adding a new initial accepting state $s$ with $s \tr{a} 0$ and $s \tr{b} 1$, and adding new transitions $n-2 \tr{a} 0$ and $n-1 \tr{b} 0$; then we determinize to get a DFA. For $S \subseteq Q_n$ and $a \in \Sig$, the transition function of the DFA is given by
\[ Sa = \begin{cases}
Sa \cup \{0\},&\text{if $n-1 \in Sa$;}\\
Sa,&\text{otherwise.}
\end{cases}
\]
We claim that the following states are reachable and pairwise distinguishable: the initial state $\{s\}$, states of the form $\{0\} \cup S$ with $S \subseteq Q_n \setminus \{0\}$, and non-empty states $S$ with $S \subseteq Q_n \setminus \{0,n-1\}$, for a total of $2^{n-1}+2^{n-2}$ states.

First consider states $\{0\} \cup S$ with $S \subseteq Q_n \setminus \{0\}$. We prove by induction on $|S|$ that all of these states are reachable. In the process, we will also show that $S$ is reachable when $\emp \ne S \subseteq Q_n \setminus \{0,n-1\}$. For the base case $|S| = 0$, note that we can reach $\{0\}$ from the initial state $\{s\}$ by $a$.

To reach $\{0\} \cup S$ with $S \subseteq Q_n \setminus \{0\}$ and $|S| > 0$, assume we can reach all states $\{0\} \cup T$ with $T \subseteq Q_n \setminus \{0\}$ and $|T| < |S|$. 
Let $q$ be the minimal element of $S$; then $1 \in Sa^{1-q}$.
More precisely, if $S = \{q,q_1,q_2,\dotsc,q_k\}$ with $1 \le q < q_1 < \dotsb < q_k \le n-1$, then $Sa^{1-q} = \{1,q_1-q+1,\dotsc,q_k-q+1\}$.
Set $T = Sa^{1-q} \setminus \{1\}$ and note that $|T| < |S|$.
By the induction hypothesis, we can reach $\{0\} \cup T$.
Apply $b$ to reach either $\{0,1\} \cup T$ (if $n-1 \in T$) or $\{1\} \cup T$ (if $n-1 \not\in T$).
Note that the only way we can have $n-1 \in T$ is if $n-1 \in S$ and $q = 1$.
Now apply $a^{q-1}$ to reach either $\{0\} \cup S$ (if $n-1 \in S$) or just $S$ (if $n-1 \not\in S$).
In the latter case, we can apply $a^{n-1}$ to reach $\{0\} \cup S$.

This shows that if $S \subseteq Q_n \setminus \{0\}$, then $\{0\} \cup S$ is reachable.
Furthermore, if $S \subseteq Q_n \setminus \{0,n-1\}$ then $S$ is reachable.

For distinguishability, if $S,T \subseteq Q_n$ and $S \ne T$, let $q$ be an element of the symmetric difference of $S$ and $T$. If $q \ne 0$ then $a^{n-1-q}$ distinguishes $S$ and $T$; if $q = 0$ use $ba^{n-2}$. To distinguish the accepting state $\{s\}$ from accepting states $S \subseteq Q_n$, use $b$.

\item
To avoid confusion between the states of $\cD_m$ and $\cD_n$, we mark the states of $\cD_m$ with primes: instead of $Q_m$ we use $Q'_m = \{0',1',2',\dotsc,(m-1)'\}$. In the restricted case, we construct an NFA for $L_m(a,b,c)L_n(a,b,c)$ by taking the disjoint union of $\cD_m(a,b,c)$ and $\cD_n(a,b,c)$, making state $(m-1)'$ non-final, and adding transitions $(m-2)' \tr{a} 0$ and $(m-1)' \tr{\sig} 0$ for $\sig \in \{b,c\}$; then we determinize to get a DFA. The states of this DFA are sets of the form $\{q'\} \cup S$, where $q' \in Q'_m$ and $S \subseteq Q_n$. For $a \in \Sig$, the transition function is given by
\[ (\{q'\} \cup S)a = \begin{cases}
\{q'a,0\} \cup Sa,&\text{if $q'a = (m-1)'$;}\\
\{q'a\} \cup Sa,&\text{otherwise.}
\end{cases}
\]
In the unrestricted case, we use the same construction with $\cD_m(a,b,c)$ and $\cD_n(a,b,c,d)$, but there are additional reachable states. In the NFA, if we are in subset $\{q'\} \cup S$, then by input $d$ we reach $S$, since $d$ is not in the alphabet of $\cD_m(a,b,c)$. So the determinization also has states $S$ where $S \subseteq Q_n$.

We claim the following states of our DFA for product are reachable and pairwise distinguishable: 
\bi
\item
\e{Restricted case:}
All states of the form $\{q'\} \cup S$ with $q' \ne (m-1)'$ and $S \subseteq Q_n$, and all states of the form $\{(m-1)',0\} \cup S$ with $S \subseteq Q_n \setminus \{0\}$.
\item
\e{Unrestricted case:}
All states from the restricted case, and all states $S$ where $S \subseteq Q_n$.
\ei
The initial state is $\{0'\}$, and we have 
\[ \{0'\} \tr{b} \{1'\} \tr{a^{m-2}} \{(m-1)',0\} \tr{a} \{1',0\} \tr{b} \{0',1\}. \]
That is, $\{0'\} \tr{ba^{m-1}b} \{0',1\}$. For $0 \le k \le n-2$ we have
$\{0',1\} \tr{a^k} \{0',1+k\}$, and $\{0',1\} \tr{c} \{0',0\}$. 
Thus all states of the form $\{0',q\}$ for $q \in Q_n$ are reachable from $\{0'\}$, using the set of words $\{x,xa,xa^2,\dotsb,xa^{n-2},xc\}$ where $x = ba^{m-1}b$. Since all of these words are permutations of $Q_n$ except for $xc$, by~\cite[Theorem 2]{Dav17} all states of the form $\{0'\} \cup S$ with $S \subseteq Q_n$ are reachable.
To reach $\{q'\} \cup S$ with $1 \le q \le m-2$, reach $\{0'\} \cup Sa^{-q}$ and apply $a^q$.
To reach $\{(m-1)',0\} \cup S$, reach $\{(m-2)'\} \cup Sa^{-1}$ and apply $a$.
In the unrestricted case, we can also reach each state $S$ from $\{0'\} \cup S$ by $d$.

To see all of these states are distinguishable, consider two distinct states $X \cup S$ and $Y \cup T$. 
In the restricted case, $X$ and $Y$ are singleton subsets of $Q'_m$; in the unrestricted case they may be singletons or empty sets.
In both cases $S$ and $T$ are arbitrary subsets of $Q_n$.
If $S \ne T$, let $q$ be an element of the symmetric difference of $S$ and $T$. If $q \ne 0$ then $a^{n-1-q}$ distinguishes the states; if $q = 0$ use $ba^{n-2}$. If $S = T$, then $X \ne Y$ and at least one of $X$ or $Y$ is non-empty. Assume without loss of generality that $Y$ is non-empty, say $Y = \{q'\}$, and assume $X$ is either empty or equal to $\{p'\}$ where $p < q$. We consider several cases:
\be[label={(\roman*)}]
\item
If $0 \not\in S$, then $a^{m-1-q}$ reduces this case to the case where $S \ne T$.
\item
If $0 \in S$ and $1 \not\in S$, and $\{p',q'\} \ne \{0',1'\}$, then $b$ reduces this to case (i).
\item
If $0,1 \in S$, and $\{p',q'\} \ne \{0',1'\}$, then $c$ reduces this to case (ii).
\item
If $\{p',q'\} = \{0',1'\}$, then $a$ reduces this to case (i), (ii) or (iii).
\ee
This shows that in both the restricted and unrestricted cases, all reachable states are pairwise distinguishable.

\item
	\be
	\item
    A binary boolean operation is \e{proper} if it depends on both arguments. For example, $\cup$, $\cap$, $\setminus$ and $\oplus$ are proper, whereas the operation $(L',L) \mapsto L$ is not proper since it depends only on the second argument.
	Since the transition semigroups of $\cD_m$ and $\cD_n$ are the symmetric groups 	$S_m$ and $S_n$, for $m,n\ge 5$, Theorem 1 of~\cite{BBMR14} applies, and all 	proper binary boolean operations have complexity $mn$. For $(m,n)\in \{(3,4),(4,3), 	(4,4)\}$ we have verified our claim by computation.
	
	\item
	This holds by \cite[Theorem 1]{BBMR14} as well.
	\item
	The upper bounds for unrestricted boolean operations on regular languages were
	derived in~\cite{BrSi17b}. 
	The proof that that the bounds are tight is very similar to the corresponding proof of  Theorem~1 in~\cite{BrSi17b}. 
For $m,n\ge 3$,  let $D'_m(a,b,-,c)$ be the dialect of $\cD'_m(a,b,c,d)$ where $c$ plays the role of $d$ and the alphabet is restricted to $\{a,b,c\}$, and let
$\cD_n(b,a,-,d)$ be the dialect of $\cD_n(a,b,c,d)$ in which $a$ and $b$ are permuted, and the alphabet is  restricted to $\{a,b,d\}$;
see Figure~\ref{fig:boolean}.

\begin{figure}[ht]
\unitlength 7.0pt
\begin{center}\begin{picture}(37,20)(-3.5,2)
\gasset{Nh=2.2,Nw=5.0,Nmr=1.25,ELdist=0.4,loopdiam=1.5}
	
\node(0')(-2,14){$0'$}\imark(0')
\node(1')(7,14){$1'$}
\node(2')(16,14){$2'$}
\node[Nframe=n](3dots')(25,14){$\dots$}
{\scriptsize
\node(m-1')(34,14){$(m-1)'$}\rmark(m-1')
}
\drawedge[curvedepth= 1.4,ELdist=-1.3](0',1'){$b$}
\drawedge[curvedepth= 1,ELdist=.3](1',0'){$b$}
\drawedge(1',2'){$a$}
\drawedge(2',3dots'){$a$}
\drawedge(3dots',m-1'){$a$}
\drawedge[curvedepth= -5.2,ELdist=-1](m-1',1'){$a$}
\drawloop(0'){$a,c$}
\drawloop(1'){$c$}
\drawloop(2'){$b,c$}
\drawloop(m-1'){$b,c$}

\gasset{Nh=2.2,Nw=5.0,Nmr=1.25,ELdist=0.4,loopdiam=1.5}
	
\node(0)(-2,7){0}\imark(0)
\node(1)(7,7){1}
\node(2)(16,7){2}
\node[Nframe=n](3dots)(25,7){$\dots$}
\node(n-1)(34,7){$n-1$}\rmark(n-1)
\drawloop(0){$b,d$}
\drawloop(1){$d$}
\drawloop(2){$a,d$}
\drawloop(n-1){$a,d$}
\drawedge[curvedepth= 1.2,ELdist=-1.0](0,1){$a$}
\drawedge[curvedepth= .8,ELdist=.25](1,0){$a$}
\drawedge(1,2){$b$}
\drawedge(2,3dots){$b$}
\drawedge(3dots,n-1){$b$}
\drawedge[curvedepth= 3.0,ELdist=-1.5](n-1,1){$b$}

\end{picture}\end{center}
\caption{Witnesses $D'_m(a,b,-,c)$ and $\cD_n(b,a,-,d)$ for boolean operations. }
\label{fig:boolean}
\end{figure}
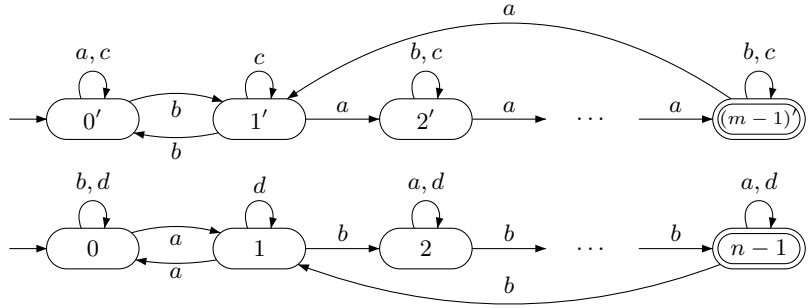

Next we complete the two DFAs by adding empty states. 
 Restricting both DFAs to the alphabet $\{a,b\}$,  leads us to the problem of determining the complexity of two DFAs over the same alphabet. 
In the direct product of the two DFAs, by \cite[Theorem 1]{BBMR14} and computation for the cases $(m,n)\in \{(3,4),(4,3),(4,4)\}$, all $mn$ states of the form $\{p',q\}$, $p'\in Q'_m$, $q\in Q_n$, are reachable and pairwise distinguishable by words in $\{a,b\}^*$ for all proper boolean operations. 
As shown in Figure~\ref{fig:cross}, 
the remaining states of the direct product are reachable; hence
all $(m+1)(n+1)$ states are reachable.

The proof of distinguishability of pairs of states in the direct product for the union, intersection and symmetric difference is the same as that in~\cite{BrSi17b}. 
The proof for difference given in~\cite{BrSi17b} is incorrect, but a corrected version is available in~\cite{BrSi16}.
\qed
	\ee
\ee
\end{proof}

\begin{figure}[h]
\unitlength 7.7pt
\begin{center}\begin{picture}(35,21)(0,-2)
\gasset{Nh=2.6,Nw=2.6,Nmr=1.2,ELdist=0.3,loopdiam=1.2}
	{\scriptsize
\node(0'0)(2,15){$0',0$}\imark(0'0)
\node(1'0)(2,10){$1',0$}
\node(2'0)(2,5){$2',0$}\rmark(2'0)
\node(3'0)(2,0){$\emp',0$}

\node(0'1)(10,15){$0',1$}
\node(1'1)(10,10){$1',1$}
\node(2'1)(10,5){$2',1$}\rmark(2'1)
\node(3'1)(10,0){$\emp',1$}

\node(0'2)(18,15){$0',2$}
\node(1'2)(18,10){$1',2$}
\node(2'2)(18,5){$2',2$}\rmark(2'2)
\node(3'2)(18,0){$\emp',2$}

\node(0'3)(26,15){$0',3$}\rmark(0'3)
\node(1'3)(26,10){$1',3$}\rmark(1'3)
\node(2'3)(26,5){$2',3$}\rmark(2'3)
\node(3'3)(26,0){$\emp',3$}\rmark(3'3)

\node(0'4)(34,15){$0',\emp$}
\node(1'4)(34,10){$1',\emp$}
\node(2'4)(34,5){$2',\emp$}\rmark(2'4)
\node(3'4)(34,0){$\emp',\emp$}
	}
\drawedge(3'0,3'1){$a$}
\drawedge(3'1,3'2){$b$}
\drawedge(3'2,3'3){$b$}
\drawedge[curvedepth=2,ELdist=.4](3'3,3'1){$b$}

\drawedge(0'4,1'4){$b$}
\drawedge(1'4,2'4){$a$}
\drawedge[curvedepth=-2,ELdist=-.9](2'4,1'4){$a$}
\drawedge(3'3,3'4){$c$}
\drawedge(2'4,3'4){$d$}

\drawedge[curvedepth=-3,ELdist=.4](0'0,3'0){$d$}
\drawedge[curvedepth=3,ELdist=.4](0'0,0'4){$c$}

\end{picture}\end{center}
\caption{Direct product for union shown partially.}
\label{fig:cross}
\end{figure}
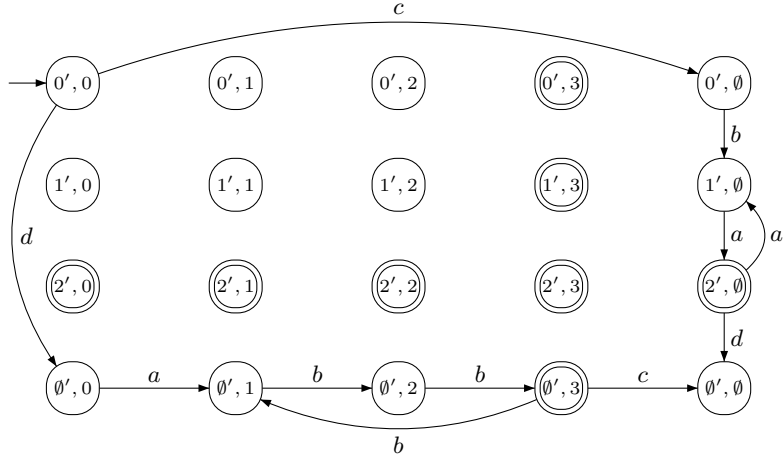

\section{Conclusions}

We have exhibited a single ternary language stream that is a witness for the maximal state complexities of star and reversal of union-free languages. Together with some dialects it also constitutes a witness for 
union, intersection, difference, symmetric difference, and product in case the alphabets of the two operands are the same. As was shown in~\cite{JiMa11} these bounds  are the same as those for regular languages.
We prove that our  witness also has the largest syntactic semigroup and most complex atoms, and that these complexities are again the same as those for arbitrary regular languages.
By adding a fourth input inducing the identity transformation to our witness we obtain witnesses for unrestricted binary operations, where the alphabets of the operands are not the same. The bounds here are again the same as those for regular languages.
In summary, this shows that the complexity measures proposed in~\cite{Brz13} do not distinguish union-free languages from regular languages.

\bibliographystyle{splncs03}
\bibliography{union-free}

\providecommand{\noopsort}[1]{}
\begin{thebibliography}{10}
\providecommand{\url}[1]{\texttt{#1}}
\providecommand{\urlprefix}{URL }

\bibitem{AfGo09}
Afonin, S., Golomazov, D.: Minimal union-free decompositions of regular
  languages. In: Dediu, A.H., et~al. (eds.) LATA 2009. LNCS, vol. 5457, pp.
  83--92. Springer (2009)

\bibitem{BBMR14}
Bell, J., Brzozowski, J.A., Moreira, N., Reis, R.: Symmetric groups and
  quotient complexity of boolean operations. In: Esparza, J., et~al. (eds.)
  ICALP 2014. LNCS, vol. 8573, pp. 1--12. Springer (2014)

\bibitem{Brz62}
Brzozowski, J.A.: Regular Expression Techniques for Sequential Circuits. Ph.D.
  thesis, Princeton University, Princeton, NJ (1962),
  \url{http://maveric.uwaterloo.ca/~brzozo/publication.html}

\bibitem{Brz13}
Brzozowski, J.A.: In search of the most complex regular languages. Int. J.
  Found. Comput. Sc.  24(6),  691--708 (2013)

\bibitem{BrDa15}
Brzozowski, J.A., Davies, S.: Quotient complexities of atoms of regular ideal
  languages. Acta Cybernet.  22,  293--311 (2015)

\bibitem{BrSi16}
Brzozowski, J.A., Sinnamon, C.: Unrestricted state complexity of binary
  operations on regular and ideal languages (2016), updated 2017.
  \url{http://arxiv.org/abs/1609.04439}

\bibitem{BrSi17b}
Brzozowski, J.A., Sinnamon, C.: Unrestricted state complexity of binary
  operations on regular and ideal languages. Journal of Automata, Languages and
  Combinatorics  22(1--3),  29--59 (2017)

\bibitem{BrSz15b}
Brzozowski, J.A., Szyku{\l}a, M.: Large aperiodic semigroups. Int. J. Found.
  Comput. Sc.  26(7),  913--931 (2015)

\bibitem{BrTa13}
Brzozowski, J.A., Tamm, H.: Complexity of atoms of regular languages. Int. J.
  Found. Comput. Sc.  24(7),  1009--1027 (2013)

\bibitem{BrTa14}
Brzozowski, J.A., Tamm, H.: Theory of \'atomata. Theoret. Comput. Sci.  539,
  13--27 (2014)

\bibitem{CDE01}
Crvenkovi\'c, S., Dolinka, I., \'Esik, Z.: On equations for union-free regular
  languages. Inform. and Comput.  164,  152--172 (2001)

\bibitem{Dav17}
Davies, S.: A new technique for reachability of states in concatenation
  automata (2017), \url{https://arxiv.org/abs/1710.05061}

\bibitem{HoKu17}
Holzer, M., Kutrib, M.: Structure and complexity of some subregular language
  families. In: Konstantinidis, S., Moreira, N., Reis, R., Shallit, J. (eds.)
  The Role of Theory in Computer Science, pp. 59--82. World Scientific (2017)

\bibitem{Iva16}
Iv\'an, S.: Complexity of atoms, combinatorially. Inform. Process. Lett.
  116(5),  356--360 (2016)

\bibitem{JiMa11}
Jir\'askov\'a, G., Masopust, T.: Complexity in union-free regular languages.
  Int. J. Found. Comput. Sc.  22(7),  1639--1653 (2011)

\bibitem{JiNa12}
Jir\'askov\'a, G., Nagy, B.: On union-free and deterministic union-free
  languages. In: Baeten, J.C.M., Ball, T., de~Boer., F.S. (eds.) TCS 2012.
  LNCS, vol. 7604, pp. 179--192. Springer (2012)

\bibitem{KuWe17}
Kutrib, M., Wendlandt, M.: Concatenation-free languages. Theoretical Computer
  Science  679(Supplement C),  83--94 (2017)

\bibitem{McPa71}
McNaughton, R., Papert, S.: Counter-Free Automata. The MIT Press (1971)

\bibitem{Nag06}
Nagy, B.: Union-free regular languages and 1-cycle-free-path-automata. Publ.
  Math. Debrecen  68(1-2),  183--197 (2006)

\bibitem{Nag10}
Nagy, B.: On union complexity of regular languages. In: CINTI 2010. pp.
  177--182. IEEE (2010)

\bibitem{Pin97}
Pin, J.E.: Syntactic semigroups. In: Handbook of Formal Languages, vol.~1:
  Word, Language, Grammar, pp. 679--746. Springer, New York, NY, USA (1997)

\bibitem{SWY04}
Salomaa, A., Wood, D., Yu, S.: On the state complexity of reversals of regular
  languages. Theoret. Comput. Sci.  320,  315--329 (2004)

\bibitem{Sch65}
Sch\"utzenberger, M.: On finite monoids having only trivial subgroups. Inform.
  and Control  8,  190--194 (1965)

\end{thebibliography}
\end{document}